\documentclass[a4paper]{amsart}
\date{July 3, 2020}

\usepackage{amsthm}
\usepackage{amsaddr}
\DeclareMathOperator{\arsinh}{\mathrm{arsinh}}

\def\const{\mathrm{const}}

\def\nz{\mathbb{N}}
\def\rz{\mathbb{R}}

\def\cE{\mathcal{E}}
\def\tf{t^{\mathrm{TF}}}
\def\cTF{\mathcal{T}^\mathrm{TF}}
\def\cW{\mathcal{T}^\mathrm{W}}
\def\cTFW{\mathcal{E}^\mathrm{TFW}}
\def\rd{\mathrm{d}}

\newtheorem{corollary}{Corollary}
\newtheorem{definition}{Definition}
\newtheorem{theorem}{Theorem}
\newtheorem{lemma}{Lemma}

\begin{document}
\author[H. Chen]{Hongshuo Chen} \address{College of
  Mathematics and Statistics, Chongqing University, Chongqing 401331,
  China} \email{hongshuo.chen@gmail.com}

\author[H. Siedentop]{Heinz Siedentop} \address{Mathematisches
  Institut\\ Ludwig-Maximilans-Universit\"at M\"unchen,
  Theresienstra\ss e 39\\ 80333 M\"unchen\\ Germany\\ and Munich Center for
  Quantum Science and Technology (MCQST)\\ Schellingstr. 4\\ 80799
  M\"unchen, Germany} \email{h.s@lmu.de}

\title[Maximal Negativity of Heavy Negative Ions]{On the Excess Charge of a
  Relativistic Statistical Model of Molecules with an Inhomogeneity
  Correction}

\begin{abstract} 
  We show that the molecular relativistic Thomas-Fermi-Weizs\"acker
  functional consisting of atoms of atomic numbers $Z_1,...,Z_k$ has a
  minimizer, if the particle number $N$ is constrained to a number
  less or equal to the total nuclear charge
  $Z:=Z_1+...+Z_K$. Moreover, there is no minimizer, if the particle
  number exceeds $2.56 Z$. This gives lower and upper bounds on the
  maximal ionization of heavy atoms.
\end{abstract}

\maketitle
\section{Introduction\label{s1}}
Shortly after the advent of quantum mechanics it became clear that the
many particle problem of interacting quantum systems cannot be solved
exactly much like in classical quantum mechanics. Thomas
\cite{Thomas1927} and Fermi \cite{Fermi1927,Fermi1928} developed a
density functional theory that turned out to describe atoms of large
atomic numbers $Z$ asymptotically correct as far as the
energy and the density on the scale $Z^{-1/3}$ is concerned (Lieb and
Simon \cite{LiebSimon1977}). Thomas and Fermi assumed in their
intuitive derivation that the potential would be locally
constant. Their functional (Lenz \cite{Lenz1932}) generalized to molecules with
atomic nuclei of atomic number $Z_k$ situated at $R_k$ reads
\begin{equation}
  \label{eq:ntf}
  \cE^\mathrm{nTF}(\rho):=  \int_{\rz^3}\rd x\left(\frac3{10}\gamma\rho^\frac{5}{3}(x)
    -\sum_{k=1}^K\frac{Z_k\rho(x)}{|x-R_k|}\right) + D[\rho] + \sum_{1\leq k<l\leq K}{Z_kZ_l\over|R_k-R_l|}
\end{equation}
where $D[\rho]:=D(\rho,\rho)$ is the quadratic form associated with
the Hermitian form
\begin{equation}
  \label{2}
  D(\rho,\sigma):= \frac12\iint_{\rz^3\times\rz^3}\rd
  x\;\rd y\;\frac{\overline{\rho(x)}\sigma(y)}{|x-y|}.
\end{equation}
The functional is naturally defined on all nonnegative
$\rho\in L^{5/3}(\rz^3)$ with finite Coulomb energy $D[\rho]$. The
positive constant $\gamma$ is physically $(3\pi^2)^{2/3}$. In the
following we will write $Z$ for the sum of the nuclear charges, i.e.,
$$Z:=Z_1+...+Z_K.$$

By Teller's lemma (Teller \cite{Teller1962}, see also Simon
\cite[Section III.9]{Simon1979}) the infimum of this functional taken
over the densities $\rho$ in the above set and all pairwise different
nuclear positions $R_1,...,R_K\in\rz^3$ is the sum of the infima of
the atomic functional (no binding) and scales in $\gamma$ and $Z_k$:
\begin{equation}
  \label{eq:tfg}
  \cE^\mathrm{nTF}(\rho)\geq -\gamma^{-1} \underbrace{e_\mathrm{TF}\sum_{k=1}^K Z_k^{7/3}}_{A:=}
\end{equation}
where $-e_\mathrm{TF}$ is the infimum of the Thomas-Fermi functional with
$K=1$, $Z=Z_1=1$, and $\gamma=1$.

Later Weizs\"acker added a correction accounting for rapidly changing
potentials. The resulting functional of the density, the
nonrelativistic TFW-functional, for atom with $q$ spin states per
electron is
  \begin{equation}
\label{energyTFDW12}
  \int_{\rz^3}\rd x\; \frac12\left|\nabla\sqrt{\rho(x)}\right|^2
    +\cE^\mathrm{nTF}(\rho).
\end{equation}
Note that 1/2 in front of the gradient term is the original constant
used by Weizs\"acker. However, there are also other constants discussed,
e.g., 1/18 emanating from the gradient expansion (Kirzhnits
\cite{Kirzhnits1957}, Hodges \cite{Hodges1973}), or 1/10 adapting the
Scott correction of the TFW functional to its physical value (Yonei
and Tomishima \cite{YoneiTomishima1965}, see also Lieb and Liberman
\cite{LiebLiberman1982} for a slightly different numerical value).

This so called inhomogeneity correction yields an exponential decay of
the atomic density as opposed the power decay of TF-theory and makes
the potential finite at the nucleus as opposed to the pure
Thomas-Fermi case which has a $|x|^{-3/2}$ singularity at the
nucleus. Also the excess charge $Q:=N-Z$, where $N$ is the maximal
number of electrons that the atom can bind, is raised from $0$ to a
positive number which Benguria and Lieb \cite{BenguriaLieb1985} bound
by a constant of order one. This reflects the experimental fact,
that for real atoms there are no doubly charged negative atomic ions.

From the physical point of view, however, the description of large
atoms with nonrelativistic theories is of limited interest, since
large atomic numbers result in velocities of the innermost electrons
that require a relativistic description. This has well known
consequence to the quantum energy (Solovej et al
\cite{Solovejetal2008}, Frank et al
\cite{Franketal2008,Franketal2009}, and Handrek and Siedentop
\cite{HandrekSiedentop2013}) and for the density (Frank et al
\cite{Franketal2019P}).

  A generalization to relativistic density functionals, however,
  suffered from the fact that the naive generalization leads to an
  energy functional that is unbounded from below and yields a
  relativistic TF-equation whose solution has necessarily infinitely
  many particles because of a $|x|^{-3}$ singularity at the origin.
  (For a review see Gombas \cite[\S 14]{Gombas1949}.)

  Dreizler and Engel \cite{EngelDreizler1987} offered a solution to
  this problem: they derived a relativistic functional from quantum
  electrodynamics and later solved it numerically for atoms (Dreizler
  and Engel \cite{EngelDreizler1988}). Following Engel and Dreizler we
  write it in terms of the Fermi momentum $p$ given by
 \begin{equation}
  \label{eq:fermimomentum}
  p(x):= (3\pi^2\rho(x))^{1/3},
\end{equation}
instead of the density $\rho$.  The TFW part (dropping the Dirac term
and and an overall trivial factor of $mc^2$) reads in the case of
molecules consisting of $K$ atoms of atomic numbers $Z_1,...,Z_K$
located at positions $R_1,...,R_K$
\begin{multline}
\label{energyrTFW}
\cTFW(p)\\
:= \cW(p)+\cTF(p)-\sum_{k=1}^K{\alpha_S Z_k\over 3\pi^2}\int_{\rz^3}\rd
x\; \frac{p^3(x)}{|x-R_k|}+{\alpha_S \over9\pi^4} D[p^3]+ \sum_{1\leq
  k<l\leq K}{\alpha_S Z_kZ_l\over|R_k-R_l|}
\end{multline}
where $\alpha_S$ is the Sommerfeld fine structure constant, which is
$1/c$ in Hartree units and has the physical value of about $1/137$.
(The atomic case is $K=1$ in which case we can assume $R_1=0$. Of
course, then $Z_1=Z$.) The Thomas-Fermi part of the kinetic energy is
\begin{equation}
  \label{eq:TF}
     \cTF(p)
    :={1\over8\pi^2}\int_{\rz^3}\rd x\; t^\mathrm{TF}(p(x)) 
  \end{equation}
  with
  \begin{equation}
    \label{eq:t}
    t^\mathrm{TF}(s):=s(s^2+1)^{3/2}+s^3(s^2+1)^{1/2}-\arsinh(s)-{8\over3}s^3.
  \end{equation}
The Weizs\"acker part of the kinetic energy is
\begin{equation}
  \label{eq:W}
  \cW(p):= {3\lambda\over8\pi^2}\int_{\rz^3}\rd x\;(\nabla p(x))^2f(p(x))^2
\end{equation}
with
\begin{equation}
  f(t):=\sqrt{\frac t{\sqrt{t^2+1}}+2\frac{t^2}{t^2+1}\arsinh(t)}.
 \end{equation}
 The constant $\lambda$ is positive. 

 We note that the ultrarelativistic limit of this functional --
 dropping the regularizing $\arsinh$ in the Weizs\"acker term -- has
 been considered before by Benguria et al \cite{Benguriaetal2008}: it
 turns out that stability of matter holds for small enough $Z$,
 whereas the functional is unbounded from below, if the $Z$ is
 large. A preliminary investigation of the massive functional appeared
 in \cite{Chen2019}.

 In this paper we consider the massive case. As opposed to the above
 case and many other relativistic models, we will show that
 $\cE^\mathrm{TFW}$ is bounded from below for all $Z_1,...,Z_K$,
 $R_1,...,R_k$ and $N$, in fact, we will show that the bound can
 be uniform in $R_1,...,R_K$ and $N$. This will be achieved by showing
 an upper bound on the number of particles that can be bounded in terms
 of the nuclear charge $Z$.

 Next we specify the space of allowed Fermi momenta $p$. To this end
 we introduce the antiderivative
 \begin{equation}
   \label{F}
   F(t):=\int_0^t\rd s\; f(s).
 \end{equation}
 We define
\begin{equation}
\label{spacep}
P:=\{p\in L^4(\rz^3)| p\geq0,\ D[p^3]<\infty,\ F\circ p\in D^1(\rz^3)\}
\end{equation}
where $D^1(\rz^3)$ is the space of all locally integrable functions on
$\rz^3$ which decay at infinity and have a square integrable gradient
(Lieb and Loss \cite[Section 8.2]{LiebLoss1996}).  For later purposes
we also introduce for $N\geq0$
\begin{equation}
  \label{spacepa}
  P_N:= \{\rho| \sqrt[3]\rho\in P,\ \int_{\rz^3}\rd x\; \rho(x)\leq N\}.
\end{equation}
It turns out that all terms occurring in the TFW functional are well
defined on $P_N$ as we shall see in the proof of Theorem \ref{thm1}. In
particular the Weizs\"acker term becomes simply
$$\tilde T^\mathrm{W}(F\circ p):= \int_{\rz^3}\rd x\; |\nabla(F\circ p)|^2(x).$$

For the proof on the excess charge, it will be convenient to write the
energy functional in terms of $\chi:= F\circ p$.

Our first results are:
\begin{theorem}[Stability]
  \label{thm1}
  For any $Z_\infty>0$, and any $K\in\nz$ there exists a constant
  $c(Z^{7/3}_\infty\cdot K)$ depending only on $Z^{7/3}_\infty\cdot K$ such for all $N\in\rz_+$,
  $Z_1,...,Z_K\in[0,Z_\infty]$, and pairwise different
  $R_1,...,R_k\in\rz^3$
  \begin{equation}
    \label{stabil} \inf \cTFW(P_N) \geq -N -c(Z^{7/3}_\infty\cdot K).
  \end{equation}
\end{theorem}
Note the surprising fact that -- unlike in many other relativistic
models of Coulomb systems -- there is no critical nuclear charge
beyond which the energy is unbounded from below. As the proof will
show this is due to the occurrence of the $\arsinh$ in the
Weizs\"acker term making it logarithmically stronger than the Coulomb
singularity. -- In Engel and Dreizler's formal derivation of the
functional from quantum electrodynamics this occurrence is a
consequence of the necessary renormalization.

Furthermore, we remark that -- unlike the case treated here -- a lower
bound which is linear in $N$ and even in $K$ -- of the
nonrelativistic TFW-functional is obvious, since it is bounded from
below by the nonrelativistic TF-functional. The stability of the
TF-functional follows then by Teller's no-binding theorem and the fact
that the excess charge of atoms is zero. The same is true
here. However, the relativistic TF kinetic energy is not strong enough
to prevent collapse: its infimum is $-\infty$ (Gombas \cite[Chapter
III, \S 14]{Gombas1949}).
\begin{theorem}[Existence of Minimizers]
  \label{thm1a}
  For any $N,Z_1,...,Z_K\geq0$, the functional $ \cTFW$ has a
  minimizer $p_N$ on $P_N$, and, moreover, if $p_N$ is a minimizer on
  $P_N$ for $N\leq Z$, then its particle number
  $\int_{\rz^3} \rd x\; p_N(x)^3/(3\pi^2)$ is equal to $N$, i.e., the
  minimizer occurs on the boundary of $P_N$.
\end{theorem}
Note that we do not claim uniqueness of the minimizer: unlike the
nonrelativistic TFW functional, the relativistic Weizs\"acker
correction is not a convex functional of the density $\rho$, i.e.,
the standard tool for showing uniqueness of minimizers is not
available.

To formulate the next theorem
we introduce the function $H$ on the positive real line by
\begin{equation}
  \label{H}
  H(s):= F(s)/(s F'(s))
\end{equation}
and write for its minimum and maximum
\begin{equation}
  \label{Hab}
  a:=\inf H(\rz_+)\ \text{and}\ b:= \sup H(\rz_+).
\end{equation}
Numerically $a=0.6116832747$.
\begin{theorem}[Bound on the Excess Charge]
  \label{thm2}
  For all $N\in\rz_+$, all
  $K\in \nz$, all $Z_1,...,Z_K\in\rz_+$, and all pairwise different
  $R_1,...,R_K\in\rz^3$ the minimizer $\rho$ of $\cTFW$ on $P_N$ fulfills
  \begin{equation}
    \int_{\rz^3}\rd x\; \rho(x)< \frac2{\sqrt a} Z.
  \end{equation}
 \end{theorem}
Note that inserting the numerical value of $a$ gives
\begin{equation}
    \int_{\rz^3}\rd x\; \rho(x)< 2.56 Z.
  \end{equation}

This result shows in particular that the relativistic TFW functional
cannot bind infinitely many electrons. Of course one should --
presumably -- regard this only as a first step, since one might
conjecture also here that the maximal excess charge $Q$ is -- like in the
nonrelativistic case -- bounded by one \cite{BenguriaLieb1985}.

Finally, we note the following immediate consequence of Theorems
\ref{thm1} and \ref{thm2}: By Theorem \ref{thm2}
$\inf_{N\in\rz_+}\inf\cTFW(P_N)\geq \inf\cTFW(P_{2a^{-1/2}Z})$. Thus, when
minimizing $\cTFW$, we might, right from the beginning, restrict to the
case that $N\leq 2 a^{-1/2} Z$. Moreover, $Z=Z_1+...+Z_K \leq Z_\infty
K$. Inserting this in \eqref{stabil} gives the following corollary: 
\begin{corollary}
  \label{t4}
  For any $Z_\infty\geq 1$, and any $K\in\nz$ there exists a constant
  $d(Z^{7/3}_\infty\cdot K)$ such for all $N\in\rz_+$,
  $Z_1,...,Z_K\in[0,Z_\infty]$, and pairwise different
  $R_1,...,R_k\in\rz^3$
  \begin{equation}
    \label{stabil}
    \inf \cTFW(P_N) \geq -d(Z^{7/3}_\infty\cdot K).
  \end{equation}
\end{corollary}
Note that physics suggests that $Z_\infty$ can
be chosen uniformly in $K$, since there are no elements known with a
nuclear charge higher than $119$, i.e., $Z_\infty=119$ is a reasonable
assumption. Assuming such a uniform bound on $Z_1,...,Z_K$ yields a
lower bound on the energy which depends only on the number of nuclei
which are present in the Coulomb system at hand.

The structure of the remaining part is basically structured according
to our three main results: In Section \ref{2.1} we prove the stability
result Theorem \ref{thm1}. The main part is a detailed lower bound on
the energy. The essential input is an estimate on the potential
energies using the nonrelativistic TF-theory and an involved estimate
of the resulting negative nonrelativistic TF kinetic energy in terms
of relativistic Weizs\"acker term and the massless relativistic TF
kinetic energy. Section \ref{Abschnitt2} contains the existence result
which is inspired by Benguria et al \cite{Benguriaetal1981}. Section
\ref{s4} and \ref{Kapitel:Beweis-Ueberschuss} show the bounds on the
ionization. Although the upper bound is inspired by Benguria's early
unpublished proof of the bound $N<2Z$ for the excess charge of the
Thomas-Fermi model, namely to integrate the Euler equation against a
suitable weight. (See also the application of this idea by Lieb
\cite{Lieb1984} for the Schr\"odinger equation and Benguria et al
\cite{Benguriaetal1992} for the Hellmann-Weizs\"acker functional.)
However, it cannot be applied in a straight forward manner in the
present context. We need to transform the functional and estimate the
resulting functional using some estimates on the function $F$ of
\eqref{F} and its inverse. We are able to control the errors so that
we loose only slightly compared to the above mentioned nonrelativistic
results \cite{Lieb1984,Benguriaetal1992}.

Finally, we collect some basic facts needed throughout the proves in
the Appendix: Appendix \ref{Anhang0} gives some results that are a
consequence of the phase space representation of the kinetic energy of
the relativistic TF-functional. Appendix \ref{Anhang1} gives the needed
basic properties of $F$.

\section{Proof of Stability\label{2.1}}
In this section we will show that the infimum of the energy functional
is bounded from below by a bound that is linear in the number of
involved electrons $N$ plus a constant depending on the number of
nuclei $K$ and the maximal occurring atomic number $Z_\infty$.

We begin with the definition of two cut-off radii $R$ and $\tilde R$
depending on parameters $\alpha$ and $\beta$. These are defined as
minimizers of two functions $F_\alpha$ and $\tilde F_{\alpha,\beta}$
\begin{definition}
  We define the function $R$:
  \begin{equation}
    \begin{split}
      R:\rz_+ &\rightarrow \rz_+\\
      \beta&\mapsto R_\beta
    \end{split}
  \end{equation}
  where, given $\beta\in\rz_+$, $R_\beta$ is the unique minimizer of
  $F_\beta$ given by
  $$ F_\beta(r):={1\over r\arsinh(r)^3} + {r\over \beta}$$
  in the variable $r\in\rz_+$.

  Furthermore, we define the function $\tilde R$:
  \begin{equation}
    \label{Rb}
    \begin{split}
      \tilde R:\rz_+^2 &\rightarrow \rz_+\\
      (\alpha,\beta)&\mapsto \tilde R_{\alpha,\beta}
    \end{split}
  \end{equation}
  where, given $(\alpha,\beta)\in\rz^2_+$, $\tilde R_{\alpha,\beta}$
is the unique minimizer of
\begin{equation}\label{Rab}
  \tilde F_{\alpha,\beta}(r):=
\begin{cases}
  {1\over r\arsinh(R_\beta)^3} + {r\over \alpha}& r\geq R_\beta\\
  \frac1{r^4}\left({R_\beta\over\arsinh(R_\beta)}\right)^3 +{r\over\alpha}&r<R_\beta
    \end{cases}
  \end{equation}
in the variable $r\in\rz_+$.   
\end{definition}
That these definitions are meaningful is a consequence of Lemma
\ref{l0} below: Because the function $F_\beta$ with $\beta\in\rz_+$
fixed, is continuous and diverges to $+\infty$ for $R\to0$ or
$R\to\infty$, its infimum on $\rz_+$ is attained. Because of the
strict convexity of $F_\beta$ its minimizer $R_\beta$ is uniquely
defined. Thus $R$ is well defined.

Since $R_\beta$ is well defined, it follows that
$\tilde F_{\alpha,\beta}$ is well defined. Now, we repeat the above
argument to define $\tilde R_{\alpha,\beta}$ again using Lemma
\ref{l0}.
\begin{lemma}
  \label{l0}
  The functions $F_\beta$ and $\tilde F_{\alpha,\beta}$ are continuous,
  diverge to $+\infty$ at $0$ and $\infty$, and are strictly convex.
\end{lemma}
\begin{proof}
  Obviously $F_\beta$ is continuous and has the stated behavior at $0$
  and $\infty$.  That it is strictly convex follows from the fact that
  it is twice differentiable (even real analytic) on $\rz_+$ and the
  second derivative is positive, since its first derivative
  \begin{equation}
    \label{Fb}
    F_\beta'(r)= -{1\over r^{2}\arsinh(r)^3}-{3\over r \arsinh(r)^4\sqrt{1+r^2}}+{1\over\beta}
  \end{equation}
  is obviously strictly increasing.

  The functions $\tilde F_{\alpha,\beta}$ are, by inspection, also
  continuous and diverge to $\infty$ at $0$ and $\infty$. Outside
  $r=R_\beta$ they are also twice differentiable (in fact again real
  analytic) and have the derivative
  \begin{equation}
    \label{Fab}
   \tilde F_{\alpha,\beta}'(r)= \begin{cases}
      -{1\over r^{2}\arsinh(R_\beta)^3}+{1\over\alpha} & r>R_\beta\\
      -{4\over r^5}\left({R_\beta\over\arsinh(R_\beta)}\right)^3 +\frac1\alpha&r<R_\beta
      \end{cases}.
    \end{equation}
    Outside $R_\beta$ these functions are monotone increasing. In
    addition at $R_\beta$ we have a jump of positive height, namely
    \begin{equation}
      \label{sprung}
      \lim_{r\nearrow R_\beta}\tilde F'_{\alpha,\beta}(r)-
      \lim_{r\searrow R_\beta}\tilde F'_{\alpha,\beta}(r)
      = {3\over R_\beta^2\arsinh(R_\beta)^3}>0,
    \end{equation}
    which shows strict
    convexity.
  \end{proof}
  We need the following basic properties of the functions $R$ and
  $\tilde R$.
\begin{lemma}
  The following properties hold:
  \begin{enumerate}
  \item For all $\beta\in\rz_+$ we have
    $R(\beta)=\tilde R(\beta,\beta)$.
  \item The functions $R$ and $\tilde R(\cdot,\beta)$ (with
    $\beta\in\rz_+$ fixed) are monotone increasing maps onto $\rz_+$.
  \end{enumerate}
\end{lemma}
\begin{proof}
  (1) For all $\beta,r\in\rz_+$ we have directly from the
  definitions of the functions $\tilde F_{\alpha,\beta}$ and $F_\beta$
  that $\tilde F_{\beta,\beta}(r)\geq F_\beta(r)$ and
  $\tilde F_{\beta,\beta}(R_\beta)=F_\beta(R_\beta)$. Thus, $R_\beta$
  does not only minimize $F_\beta$ but also $\tilde
  F_{\beta,\beta}$. However, the minimizer of $\tilde F_{\beta,\beta}$
  is uniquely determined. Thus, $\tilde R_{\beta,\beta}=R_\beta$.

  (2) $R(\beta)$ solves the equation $F_\beta'(r)=0$. From \eqref{Fb}
  it is obvious that the function $R$ is monotone increasing. It is
  also obvious from \eqref{Fb} that $R_\beta\to0$ as $\beta\to0$ and
  $R_\beta\to \infty$ as $r\to\infty$. This shows the claim on $R$.

  By strict convexity, $\tilde F_{\alpha,\beta}'(r)<0$ for
  $r<\tilde R_{\alpha,\beta}$ and $\tilde F_{\alpha,\beta}'(r)>0$ for
  $r>\tilde R_{\alpha,\beta}$ holding for all $\alpha,\beta\in\rz_+$,
  holding in particular for $\alpha=\beta$. Moreover, for
  $\alpha>\beta$ and $r<R_\beta$ we have
  \begin{equation}
    \tilde F_{\alpha,\beta}'(r) =- {4\over r^5}\left({R_\beta\over\arsinh(R_\beta)}\right)^3+\frac1\alpha < F_{\beta,\beta}'(r)<0,
  \end{equation}
  i.e., $\tilde R_{\alpha,\beta}\geq R_\beta$ for $\alpha>\beta$. Away
  from $R_\beta$ we can use again the Euler equation: suppose that
  there is a $\alpha'\in\rz_+$ such that
  $\tilde R_{\alpha',\beta}>R_\beta$. Then $\tilde R_{\alpha',\beta}$
  fulfills
  $$\tilde F_{\alpha',\beta}'(\tilde R_{\alpha',\beta}) =-{1\over \tilde R_{\alpha',\beta}^2\arsinh(R_\beta)^3}+\frac1{\alpha'}=0.$$
  Thus, we have for any $\alpha''>\alpha'$, an
  $r'>\tilde R_{\alpha',\beta}>R_\beta$ solving the equation
  $$\tilde F_{\alpha'',\beta}'(r)=-{1\over r^2\arsinh(R_\beta)^3}+\frac1{\alpha''}=0.$$
  Since $\tilde R_{\alpha'',\beta}>R_\beta$, we have $r'=\tilde R_{\alpha'',\beta}$. Thus $\tilde R_{\alpha'',\beta}> \tilde R_{\alpha',\beta}$ for
  $\alpha''>\alpha'>\beta$. This implies  the monotony in $\alpha$
  for $\alpha>\beta$. Similar arguments yield the monotony in $\alpha$
  for $\alpha\leq\beta$.  
\end{proof}

\begin{proof}[Proof of Theorem \ref{thm1}]
The condition $\nabla(F\circ p)\in L^2(\rz^3)$ is a mere
rewriting of the finiteness condition of the Weizs\"acker
term. Similarly $D[p^3]<\infty$ is equivalent to the finiteness of the
electron-electron repulsion.

Next we look at the Thomas-Fermi term. Obviously the massive
Thomas-Fermi term is bounded above by the massless one. It is also
bounded from below by the massless Thomas-Fermi term minus the particle number
(see \eqref{eq:masselos}), i.e.,
\begin{equation}
  \label{tfmasselos}
  \cTF(p)\geq  {1\over 4\pi^2}\int_{\rz^3}\rd x\; p(x)^4 - N.
\end{equation}
Both bounds are obvious from the representation of $\cTF$ as phase space
integral.

Since $p\in L^4(\rz^3)$, $p$ decays at infinity and therefore also
$F\circ p$. Thus, we may employ the Sobolev inequality yielding
$F\circ p\in L^6(\rz^3)$ and estimate $F$ from below using \eqref{a1}
from the appendix. We get
\begin{multline}
  \label{sw}
  \int_{\rz^3}\rd x\; |\nabla(F\circ p)(x)|^2
  \geq c_s \left(\int_{\rz^3}\rd x\;|F\circ p (x)|^6\right)^{1/3}\\
  \geq \frac{c_s}4\left(\int_{\rz^3}\rd x\;p(x)^6\arsinh(p(x))^3\right)^{1/3} =:{c_s\over4} W(p)^{1/3}
\end{multline}
where $c_s$ is the Sobolev constant.
Thus, using \eqref{sw}, \eqref{eq:tfg}, and \eqref{tfmasselos} we get
\begin{equation}
  \label{eq:lobo1}
  \cTFW(p)\geq {3c_s\lambda\over32\pi^2}W(p)^{1/3}+{1\over4\pi^2}T(p)-N -\frac3{10}\gamma\alpha_S \int_{\rz^3}\rd x\; p(x)^5 -\frac{\alpha_S}\gamma A
\end{equation}
with
$T(p):= \int_{\rz^3}\rd x\;p(x)^4$. We pick
$$\gamma=\sqrt{A\over \tfrac3{10}\int_{\rz^3}\rd x\; p(x)^5}$$
and get
\begin{equation}
  \label{eq:lowbo2}
  \cTFW(p)\geq {3\lambda c_s\over32\pi^2}W(p)^{1/3}+{1\over4\pi^2} T(p)-N
  -2\alpha_S\sqrt{\frac{3A}{10}\int_{\rz^3}\rd x\; p(x)^5 }.
\end{equation}

We pick $\beta\in\rz_+$ to be specified later but independently of $p$
and estimate
\begin{eqnarray}
  &&\int_{\rz^3}\rd x\; p(x)^5\\
  &\leq& {1\over r \arsinh(r)^3}\int_{p(x)>r}\rd x\; p(x)^6\arsinh(p(x))^3 + r\int_{p(x)\leq r}\rd x\;p(x)^4\label{15}\\
  &\leq & {1\over r \arsinh(r)^3} W(p) + r T(p)\label{16}\\
  &\leq& \label{17}
         \begin{cases}
    {1\over r \arsinh(R_\beta)^3} W(p) + r T(p)& r\geq R_\beta\\
    \frac1{r^4} W(p)(R_\beta/\arsinh(R_\beta))^3 + r T(p) & r<R_\beta
    \end{cases}
\end{eqnarray}
where $R_\beta$ is the unique minimizer of the function $F_\beta$
defined in \eqref{Rb}. Using \eqref{Rab} allows us to rewrite
\eqref{17} so that we get
\begin{eqnarray}
  \label{eq:28a}
  \int_{\rz^3}\rd x\; p(x)^5 \leq W(p) \tilde F_{Q,\beta}(r)
\end{eqnarray}
in the variable $r$. Here we use the abbreviation $Q:=W(p)/T(p)$. The resulting minimizer
$\tilde R_{Q,\beta}$ exists uniquely for each $p$ and $\beta$
by Lemma \ref{l0}. Inserting it into \eqref{eq:28a} gives
\begin{eqnarray}
  \label{eq:28b}
  &\int_{\rz^3}\rd x\; p(x)^5 \leq W(p) \tilde F_{Q,\beta}(\tilde R_{Q,\beta})\\
  =  &\begin{cases}
    {1\over \tilde R_{Q,\beta} \arsinh(R_\beta)^3} W(p) + \tilde R_{Q,\beta}T(p) & Q\geq \beta\\
    \frac1{R_{Q,\beta}^4} W(p)(R_\beta/\arsinh(R_\beta))^3 + R_{Q,\beta} T(p) & Q<\beta
    \end{cases}.
\end{eqnarray}
In the last step we used that $\tilde R_{\alpha,\beta}$ is monotone
increasing in $\alpha$ and $\tilde R_{\beta,\beta}=R_\beta$.

The optimizer of the first line is
$$ \tilde R_{Q,\beta}=\sqrt{W(p)\over T(p)\arsinh(R_\beta)^3}.$$
The second line is minimized for
$$\tilde R_{Q,\beta}= \sqrt[5]{4W(p)R_\beta^3\over\arsinh(R_\beta)^3T(p)}.$$
Thus, 
\begin{eqnarray}\nonumber
  \int_{\rz^3}\rd x\; p(x)^5 &\leq&
  \begin{cases}
       2 \sqrt{W(p)T(p)\over\arsinh(R_\beta)^3}  &W(p)\geq \beta T(p)\\
      \frac5{4^{4/5}}(W(p)(R_\beta/\arsinh(R_\beta)^3)^{1/5}T(p)^{4/5}& W(p)<\beta T(p)   
    \end{cases}\\
     &\leq&
     \begin{cases}
      2 \sqrt{W(p)T(p)\over\arsinh(R_\beta)^3}  &W(p)\geq \beta T(p)\\
      \frac5{4^{4/5}}(\beta(R_\beta/\arsinh(R_\beta)^3)^{1/5}T(p)& W(p)<\beta T(p)   \label{21}
     \end{cases}.
  \end{eqnarray}
  We insert this bound in \eqref{eq:lowbo2} and obtain
  \begin{multline}
    \label{lowbo3}
    \cTFW(p)
    \geq{3\lambda c_s\over32\pi^2} W(p)^{1/3}+\frac1{4\pi^2}T(p)-N\\
    -    2\alpha_S
    \begin{cases}
      \sqrt{{3A \over5\sqrt{2\arsinh(R_\beta)^3}}}\left({W(p)^{1/3}\over \frac43}+{T(p)\over4}\right)  &W(p)\geq \beta T(p)\\
      \sqrt{\frac32{A\over 4^{4/5}}(\beta R_\beta^3/\arsinh(R_\beta)^3)^{1/5}
        T(p)}&
      W(p)<\beta T(p).
    \end{cases}
    \end{multline}
    Now, we pick $\beta$ such that
    \begin{equation}
      \label{eq:optbeta}
      \min\{\frac{\lambda c_s}{8\pi^2},\frac1{\pi^2}\}
      =\sqrt{{3A \over5\sqrt{2\arsinh(R_\beta)^3}}}.
    \end{equation}
    This bounds the case of the first line from below by $-N$.

    The second line is bounded from below -- independently of $p$, since
    the leading power in $T(p)$ has a positive coefficient.  To make
    this quantitative, we write $a$ for the coefficient of $T(p)$ and
    $b$ for the coefficient of $T(p)$ in the second case of
    \eqref{lowbo3}.  Now $aT(p)-b\sqrt{T(p)}$ is minimized for
    $T(p)= b^2/(4a^2)$ implying
    \begin{equation}
      \label{optT}
      aT(p)-b\sqrt{T(p)}\geq -b^2/(4a)=:C(A)
    \end{equation}
    independently of $p$.

    Since \eqref{eq:optbeta} implies that $R_\beta$ is increasing in
    $A$, we find that $C(A)$ is increasing. Moreover by the definition
    of $A$ in \eqref{eq:tfg}, we can estimate
    $A\leq e_\mathrm{TF} Z_\infty^{7/3}\cdot K$. Thus, the second case
    of \eqref{lowbo3} is bounded from below by
    $-N - C(e_\mathrm{TF} Z_\infty^{7/3}\cdot K)$. Writing the latter
    constant as $c(Z^{7/3}_\infty\cdot K)$ yields \eqref{stabil}.
  \end{proof}
  
\section{Proof of the Existence of Minimizers (Theorem
  \ref{thm1a}\label{Abschnitt2})}
    
      Since $\cTFW$ is bounded from below on $P_N$, we will now address
    the question whether the infimum is attained. It is convenient, to
    regard the functional as function of the density $\rho$ instead of
    the Fermi momentum $p$ and similarly for other parts of the energy
    functional. In abuse of notation we write $\cTFW(\rho)$ instead of
    $\cTFW(\sqrt[3]{3\pi^2\rho})$, i.e.,
    \begin{align}
      \label{tfwr}
      &\cTFW(\rho)\nonumber\\
      =& {3\lambda\over8\pi^2}\int_{\rz^3}\rd x\;|\nabla( F\circ\sqrt[3]{3\pi^2\rho})(x)|^2
          + \cTF(\rho)- \alpha_S \sum_{k=1}^K\int_{\rz^3}\rd x{Z_k\rho(x)\over|x-R_k|} 
          +\alpha_S D[\rho]\\
      &+\sum_{1\leq k<l\leq K}{\alpha_S Z_kZ_l\over |R_k-R_l|}.\nonumber
    \end{align}
    \begin{theorem}
\label{thmmin}
For every $N,Z_1,...,Z_K\in\rz_+$ and $R_1,...,R_K\in\rz^3$ there
exists $\rho\in P_N$ such that
\begin{equation}
\label{Gleichung.6.0}
\begin{aligned}
\cTFW(\rho)=\inf \cTFW(P_N).
\end{aligned}
\end{equation}
\end{theorem}

  Since $\cTFW$ has a lower bound, there is a minimizing
  sequence {$\rho_j$}, such that
\begin{equation}
\label{Gleichung.6.1}
\begin{aligned}
\lim_{j\to\infty}\cTFW(\rho_j)=\inf\cTFW(P_N).
\end{aligned}
\end{equation}
However \eqref{lowbo3} shows not only boundedness from below but also
that $\cTFW(\rho_j)\to\infty$ as either of the
norms $\|\rho_j\|_{L^{4/3}(\rz^3)}$,
$\|F\circ \sqrt[3]{3\pi^2\rho_j} \|_{D^1(\rz^3)}$, or
$\|\rho_j\|_C:=\sqrt{D[\rho]}$ tend to infinity. By the Banach-Alaoglu
theorem we can pick a subsequence, such that we have weak convergence
in all of these norms. We now imagine that we started already with
this subsequence to avoid a new notation. We have
$\chi\in D^1(\rz^3)$,  $\zeta_2\in L^{4/3}(\rz^3)$, 
$\zeta_3$ with $D[\zeta_3]<\infty$, and $\zeta_4\in L^{5/3}(\rz)$ such that, as $j\to\infty$,
  \begin{eqnarray}
    \label{eq:b1}
    \int_{\rz^3}\rd x\; \nabla f(x) \nabla (F\circ \sqrt[3]{3\pi^2\rho_j} -\chi)(x)\to 0&& \forall f\in D^1(\rz^3),\\
    \int_{\rz^3}\rd x\; f(x) (\rho_j-\zeta_2)(x)\to 0&& \forall f\in L^4(\rz^3),\\
    D(f,\rho_j-\zeta_3)\to 0 && \forall f\ \text{with}\ D[f]<\infty,\\
    \label{eq:b4}
    \int_{\rz^3}\rd x\; f(x)(\rho_j-\zeta_4)(x)\to 0&& \forall f\in L^{5/2}(\rz^3)
  \end{eqnarray}
  using the abbreviation \eqref{2}.  The convergence \eqref{eq:b4}
  holds, since by \eqref{21}, $\|\rho_j\|_{L^{5/3}(\rz^3)}$ is also
  bounded.
  
To prove the existence of the minimizer, we prove the lower
semicontinuity of each term of $\cTFW$.

We begin with the Weizs\"acker term. Since the norm on
$D^1(\rz^3)$ is lower continuous (Lieb and Loss \cite[Section
8.2]{LiebLoss1996}), we immediately have
  \begin{equation}
    \label{halbstetigW}
    \cW(\zeta) \leq \liminf_{j\to\infty}\cW(\rho_j)
  \end{equation}
  with $\zeta:= (3\pi^2)^{-1}(F^{-1}(\chi))^3$.

  Next we consider the $\cTF$ which obviously is a convex
  functional with a derivative $d$ at $\rho$ (see also Appendix \ref{Anhang0})
  \begin{eqnarray}
    \label{eq:ableitung}
    [d(\cTF)(\rho)]: L^{4/3} &\to& \rz\\
    \eta &\mapsto & \int_{\rz^3} \rd x\; (\sqrt{(3\pi^2\rho)^{2/3}+1}-1) \eta(x).
  \end{eqnarray}
  Since $d(\cTF)(\rho)$ is in the dual space of $L^{4/3}$. Therefore, we get by convexity
  \begin{equation}
    \liminf_{j\to\infty}\cTF(\rho_j)\geq \cTF(\zeta_2) +\liminf_{j\to\infty}d(\cTF)(\eta_2)(\rho_j-\zeta_2)=\cTF(\zeta_2).
  \end{equation}

  The third term, the external potential, is actually weakly
  continuous at the limiting element. Because of linearity and
  translational invariance it suffices to show this for the Coulomb
  potential of a unit point charge at the origin. We decompose it into
  a singular and a long range part
  \begin{equation}
    \label{3.term}
    1/|x| := l(x) + r(x)
  \end{equation}
  with
  $$r(|x|) : = \int_{|y|<1}\rd y\;{\mu(y)\over|x-y|},$$
  where we pick some spherically symmetric charge distribution
  supported in the unit ball centered at the origin with
  $D[\mu]<\infty$. Whereas $l\in L^{5/2}$, since
  $\mathrm{supp}(l)\subset B_1(0)$ and $l(x)\leq 1/|x|$ because of
  Newton's theorem (see \cite[Section III.9]{Simon1979}).

  Thus
  $$\lim_{j\to\infty}\int_{\rz^3}\rd x\; r(x)\rho_j(x) = 2\lim_{j\to\infty}D(\mu,\rho_j)=2D(\mu,\zeta_3) = \int_{\rz^3}\rd x\;r(x)\zeta_3(x).$$
  Since $l\in L^{5/2}(\rz^3)$ we have
  $$\lim_{j\to\infty}\int_{\rz^3}\rd x\;l(x)\rho_j(x) =\lim_{j\to\infty}\int_{\rz^3}\rd x\;l(x)\zeta_4(x).$$

  Eventually we consider the electron-electron repulsion. By the
  Schwarz inequality
  $$D(\zeta_3,\zeta_3) = \lim_{j\to\infty}D(\rho_j,\zeta_3)\leq \liminf_{j\to\infty}\sqrt{D[\rho_j]D[\zeta_3]}$$
  and thus
  $$ D[\zeta_3]\leq \liminf_{j\to\infty} D[\rho_j].$$
 
  Thus all the terms are lower continuous at the corresponding
  limiting points. To conclude the proof, we wish to show that
  $\zeta=\zeta_2=\zeta_3=\zeta_4$ and that
  the limiting point is in $P_N$.

  Since $C_0^\infty(\rz^3)\subset L^4(\rz^3)\cap L^{5/2}(\rz^3)$ we
  have for all $f\in C_0^\infty(\rz^3)$ the two equalities
  $$ \int_{\rz^3}\rd x\; f(x)\zeta_2=\lim_{j\to\infty}\int_{\rz^3}\rd x\; f(x) \rho_j(x)=\int_{\rz^3}\rd x\; f(x)\zeta_4,$$
  i.e., $\zeta_2=\zeta_4$.

  Next we take $-(4\pi)^{-1}\Delta f$ with
  $f\in C_0^\infty(\rz^3)$. Then
  \begin{multline}
    \int_{\rz^3}\rd x\;f(x)\zeta_3(x)=D(-(4\pi)^{-1}\Delta f, \zeta_3)
    =\lim_{j\to\infty}D(-(4\pi)^{-1}\Delta f, \rho_j)\\
    =    \lim_{j\to\infty}\int_{\rz^3}\rd x\; f(x) \rho_j(x)=
    \int_{\rz^3}\rd x\;f(x)\zeta_2(x).
    \end{multline}
  Thus $\zeta_2=\zeta_3$. 

  By \cite[Theorem 8.7]{LiebLoss1996} we can also assume that
  $\chi_j:=F(\sqrt[3]{3\pi^2\rho_j})\to\chi|_A$ converges in
  $L^{6}(A)$ strongly on sets $A\subset\rz^3$ of finite measure. We
  will show now, that $\rho_j$ converges strongly in $L^{4/3}(A)$ to
  $\zeta|_A$ and thus $\zeta=\zeta_2$, since $A$ is arbitrary. To do
  this we will use an estimate on the inverse of $F$ that follows from
  Lemma \ref{la1} of Appendix \ref{Anhang1}.
  \begin{multline}
    (3\pi^2)^{4/3}\int_A\rd x\;|\rho_j(x)-\zeta(x)|^{4/3}
    = \int_A\rd x\;\left|(F^{-1}(\chi_j(x)))^3-(F^{-1}(\chi(x)))^3\right|^{4/3}\\
    \underset{\text{by}\ \eqref{Ginverse}}\leq 8 \int_A\rd
    x\;\left|\chi_j(x)^2(\chi_j(x)+1)-\chi(x)^2(\chi(x)+1)\right|^{4/3}\\
    \leq 8 \int_A\rd
    x\;\left|(\chi_j(x)-\chi(x))(\chi_j(x)^2+\chi_j(x)\chi(x)+\chi(x)^2+\chi_j(x)+\chi(x))\right|^{4/3}\\
    \leq8 \left(\int_A\rd x\; \left|\chi_j(x)-\chi(x)\right|^{9/2}\right)^{2/9}
      \left(\int_A\rd x\; \left|\chi_j^2+\chi_j\chi+\chi^2+\chi_j+\chi\right|^{12/7}\right)^{7/9}
  \end{multline}
  which tends to zero, since the last factor is uniformly bounded in
  $j$.

  Finally, $\zeta\in P_N$, since otherwise this would lead to
  immediate contradictions to $\zeta \geq 0$ or
  $\int_{\rz}\rd x\;\zeta(x)\leq N$.
 
  \section{Proof of a Lower Bound on the Maximal Ionization\label{s4}}

  \begin{proof}
    We will now assume that $N\leq Z$ and -- contradictory to the
    assumption of the theorem that the minimizer has a particle number
    that is strictly less than $N$, i.e.,
    $\int_{\rz^3} \rd x\;p_N(x)^3/3\pi^2 < N$. If this is the case, then $p_N$
    fulfills the Thomas-Fermi-Weizs\"acker equation
    \begin{equation}
      \label{Gleichung.7.3.0}
      \begin{aligned}
        &\frac{8p_N^2(x)(\sqrt{p_N^2(x)+1}-1)}{F'(p_N)(x)}-6\lambda\Delta F(p_N)(x)-\sum_{k=1}^K\frac{8\alpha_S Z_k p_N^2(x)}{|x-R_k|F'(p_N)(x)}\\
        &+\frac{8\alpha_S}{3\pi^2}\int_{\rz^3}\rd
        y\;\frac{p_N^2(x)p_N^3(y)}{|x-y|F'(p_N)(x)}=0.
      \end{aligned}
    \end{equation}
    The following is inspired by an idea of Benguria, Br{\'e}zis, and
    Lieb \cite{Benguriaetal1981}. We choose $\zeta_0\in C_0^\infty$
    the same as in Benguria, Br{\'e}zis, and Lieb
    \cite{Benguriaetal1981}. It is a spherically symmetric function
    such that $\mathrm{supp}(\zeta_0)\subset B_2(0)\setminus
    B_1(0)$. Set $\zeta_n(x)=\zeta_0(x/n)$. By (\ref{Gleichung.7.3.0})
    we have,
    \begin{multline}
      \label{Gleichung.7.3.1}
      \int_{\rz^3}\rd x\;\zeta_n^2(x)\left(8(\sqrt{p_N^2(x)+1}-1)-\frac{6\lambda F'(p_N)(x)\Delta F(p_N)(x)}{p_N^2(x)}\right. \\
      -\sum_{k=1}^K\frac{8\alpha_S
        Z_k}{|x-R_k|}\left.+\frac{8\alpha_S}{3\pi^2}\int_{\rz^3}\rd
        y\;\frac{p_N^3(y)}{|x-y|}\right)=0.
    \end{multline}
    Integrating by parts and using the Schwarz inequality, we have
    \begin{equation}
      \label{Gleichung.7.3.2}
      \begin{aligned}
        &-\int_{\rz^3}\rd x\;\zeta_n^2\frac{F'(p_N)\Delta F(p_N)}{p_N^2}\\
        =&\int_{\rz^3}\rd x\;\left(2\zeta_n\nabla\zeta_n\frac{F'(p_N)}{p_N^2}+\zeta_n^2\left(\frac{F'(p_N)}{p_N^2}\right)'\nabla p_N\right)F'(p_N)\nabla p_N\\
        \le&\frac{1}{\epsilon}\int_{\rz^3}\rd x\;|\nabla\zeta_n|^2+\int_{\rz^3}\rd x\;\left(\epsilon\left(\frac{F'^2(p_N)}{p_N^2}\right)^2+\left(\frac{F'(p_N)}{p_N^2}\right)'F'(p_N)\right)\zeta_n^2|\nabla p_N|^2.\\
      \end{aligned}
    \end{equation}
    Using the definition of $F(p_N)$, we get
    \begin{equation}
      \label{Gleichung.7.3.3}
      \begin{aligned}
        \left(\frac{F'(p_N)}{p_N^2}\right)' =
        -\frac{(2p_N^2+3)\sqrt{p_N^2+1}+4(2p_N^2+1)p_N
          \arsinh(p_N)}
        {2p_N^\frac{5}{2}(p_N^2+1)^\frac{3}{2}\sqrt{\sqrt{p_N^2+1}+2p_N
            \arsinh(p_N)}}<0.
      \end{aligned}
    \end{equation}
    Define
    \begin{equation}
      \label{Gleichung.7.3.4}
      \begin{aligned}
        g(s):=-\frac{\left(\frac{F'(s)}{s^2}\right)'F'(s)}{\left(\frac{F'^2(s)}{s^2}\right)^2},\ \ g(p_N)>0.
      \end{aligned}
    \end{equation}
    We easily get $g(0_+)=3/2$ and
    $\lim\limits_{s\rightarrow\infty}g(s)=\infty$. So
    $c_g:=\min\limits_{s\ge 0}g(s)>0$. Choose $\epsilon=c_g$. Then
    \begin{equation}
      \label{Gleichung.7.3.5}
      \begin{aligned}
        -\int_{\rz^3}\rd x\;\zeta_n^2\frac{F'(p_N)\Delta
          F(p_N)}{p_N^2}\le\frac{1}{c_g}\int_{\rz^3}\rd x\;|\nabla\zeta_n|^2\le
        Cn.
      \end{aligned}
    \end{equation}
    Next, we compute
    \begin{equation}
      \label{Gleichung.7.3.6}
      \begin{aligned}
        &\int_{\rz^3}\rd x\;\zeta_n^2(\sqrt{p_N^2+1}-1)<\int_{\rz^3}\rd
        x\;\zeta_n^2p_N\le\epsilon_nn^2,
      \end{aligned}
    \end{equation}
    where $\epsilon_n\rightarrow0$ as $n\rightarrow\infty$. The last
    inequality is proved by Benguria, Br{\'e}zis, and Lieb
    \cite{Benguriaetal1981}. About the external potential term, since
    $\zeta_n(x)=0$ for $|x|<n$, we have
    \begin{equation}
      \label{Gleichung.7.3.7.0}
      \begin{aligned}
        &\int_{\rz^3}\rd x\;\zeta_n^2(x)\sum_{k=1}^K\frac{8\alpha_S Z_k}{|x-R_k|}\\
        =&\int_{\rz^3}\rd x\;\zeta_n^2(x)\sum_{k=1}^K8\alpha_S Z_k\left(\frac{1}{|x|}+\frac{|x|-|x-R_k|}{|x||x-R_k|}\right)\\
        \ge&\int_{\rz^3}\rd x\;\zeta_n^2(x)\sum_{k=1}^K8\alpha_S Z_k\left(\frac{1}{|x|}-\frac{\max \{|R_k|\}}{|x|(n-\max \{|R_k|\})}\right)\\
        \ge&\left(1-\frac{c}{n}\right)8\alpha_S\sum_{k=1}^K Z_k\int_{\rz^3}\rd x\;\frac{\zeta_n^2(x)}{|x|}\\
      \end{aligned}
    \end{equation}
    for large $n$. We now address the remaining term: since $\zeta_n$
    is spherically symmetric, by a result of Lieb and Simon
    \cite[Eq. (35)]{LiebSimon1977} we have
    \begin{equation}
      \label{Gleichung.7.3.7.1}
      \begin{aligned}
        &\int_{\rz^3}\rd x\;\zeta_n^2(x)\frac{8\alpha_S}{3\pi^2}\int_{\rz^3}\rd y\;\frac{p_N^3(y)}{|x-y|}\\
        =&\int_{\rz^3}\rd x\;\zeta_n^2(x)\left[\frac{8\alpha_S}{3\pi^2}\int_{\rz^3}\rd y\;\frac{p_N^3(y)}{|x-y|}\right]\\
        =&\int_{\rz^3}\rd x\;\zeta_n^2(x)\frac{8\alpha_S}{3\pi^2}\int_{\rz^3}\rd y\;\frac{p_N^3(y)}{\max (|x|,|y|)}\\
        \le&8\alpha_S N\int_{\rz^3}\rd x\;\frac{\zeta_n^2(x)}{|x|},\\
      \end{aligned}
    \end{equation}
    where $[\varphi]$ denotes the spherical average of $\varphi$,
    i.e.,
    $$[\varphi](x)=\frac{1}{4\pi}\int_{\mathbb{S}^2}\rd
    \Omega\;\varphi(|x|\Omega).
    $$
    Thus, for large $n$, we find
    \begin{equation}
      \label{Gleichung.7.3.7}
      \begin{aligned}
        &\int_{\rz^3}\rd x\;\zeta_n^2(x)\left(-\sum_{k=1}^K\frac{8\alpha_S
            Z_k}{|x-R_k|}+\frac{8\alpha_S}{3\pi^2}\int_{\rz^3}\rd
          y\;\frac{p_N^3(y)}{|x-y|}\right)\le c(N-Z)n^2+cn.
      \end{aligned}
    \end{equation}
    Combining (\ref{Gleichung.7.3.1}), (\ref{Gleichung.7.3.5}),
    (\ref{Gleichung.7.3.6}), and (\ref{Gleichung.7.3.7}), we find
    \begin{equation}
      \epsilon_nn^2+Cn+c(N-Z)n^2\ge 0.
    \end{equation}
    As $n\rightarrow\infty$, we have that $Z\le N$ which contradicts
    the assumption that $N<Z$.
 \end{proof}
 
\section{Proof of the Upper Bound on the Maximal
  Ionization \label{Kapitel:Beweis-Ueberschuss}}
\begin{proof}[Proof of $N\leq \const Z$]
It will be convenient to express the TFW functional in terms of
\begin{equation}
\label{psi}
\chi:=F\circ p
\end{equation}
which is guided by the idea to make the dominating term in the energy
simple. Obviously $F$ is strictly monotone and $F(\rz_+)=\rz_+$. We get
\begin{multline}
  \label{echi}
     \tilde \cE^{\mathrm{TFW}}(\chi):=\cE^\mathrm{TFW}(F^{-1}\circ\chi)\\
    = {3\lambda\over8\pi^2}\int_{\rz^3}\rd x\; |\nabla\chi(x)|^2 +
    \frac{1}{8\pi^2}\int_{\rz^3}\rd x\;t^\mathrm{TF}\circ F^{-1}\circ\chi(x)\\
    -\sum\limits_{k=1}^K{\alpha_S Z_k\over3\pi^2}\int_{\rz^3}\rd x\;
    {(F^{-1}(\chi(x)))^3\over|x-R_k|}+{\alpha_S\over9\pi^4}
    D[(F^{-1}\circ\chi)^3]+ \sum_{1\leq k<l\leq K}{\alpha_S Z_kZ_l\over|R_k-R_l|}.
   \end{multline}

Suppose that $\chi$ minimizes the TFW functional on $F(P_N)$. Then it satisfies
\begin{multline}
  \label{eq:euler}
  -{3\lambda\over4\pi^2}\Delta\chi   +\frac1{8\pi^2}T'(F^{-1}(\chi))(F^{-1})'(\chi)\\
  -\sum\limits_{k=1}^K{Z_k\alpha_S \over \pi^2}{(F^{-1})'(\chi)(F^{-1}(\chi))^2\over|x-R_k|} +
  {\alpha_S\over3\pi^4}
  (F^{-1})'(\chi)(F^{-1}(\chi))^2\int_{\rz^3}\rd y{(F^{-1}(\chi(y)))^3\over|x-y|} =0.
\end{multline}
Following the spirit of Lieb \cite{Lieb1984}, we multiply by
$\chi/\phi(x)$ and integrate. The function $\phi$ is defined as
\begin{equation}
	\label{phi}
	\phi(x):=\sum\limits_{k=1}^K\frac{\kappa_k}{|x-R_k|}.
\end{equation}
The coefficients $\kappa_k>0$ will be given later. Due to our
transform the Weizs\"acker term becomes easy. We can follow Lieb's
argument \cite[Eq. (3.17)]{Lieb1984} step by step to see that it is
positive.

Since $F$ is the antiderivative of a positive expression, $F'$ is
positive and thus also the derivative of $F^{-1}$. Therefore
$$ \int_{\rz^3} \rd x\;\frac{\chi(x)}{\phi(x)} {t^\mathrm{TF}}'(F^{-1}(\chi))(F^{-1})'(\chi(x))>0$$
when $\chi$ does not vanish almost everywhere.

Thus we have the inequality
\begin{multline}
  \label{ungleichung-p}
  -\sum\limits_{k=1}^K Z_k\int_{\rz^3}\rd x\;\frac{\chi(x)}{\phi(x)|x-R_k|} (F^{-1})'(\chi(x))(F^{-1}(\chi(x)))^2\\
  +{1\over3\pi^2}\int_{\rz^3}\rd x\; \int_{\rz^3}\rd y\; { \chi(x)((F^{-1})'(\psi(x)))^2(F^{-1}(\chi(y)))^3\over{\phi(x)|x-y|}} < 0
\end{multline}
where we assumed that $\chi$ is not vanishing almost everywhere, which
we may, since otherwise the claim is trivial. Rewriting this in $p$ it
becomes
\begin{multline}
  \label{eq:ungleichung}
  -3\pi^2 \sum\limits_{k=1}^K Z_k\int_{\rz^3}\rd x\; \frac{F(p(x))/F'(p(x))p(x)^2}{\phi(x)|x-R_k|}\\
  +\int_{\rz^3}\rd x\; \int_{\rz^3}\rd y\; {F(p(x))/F'(p(x))p(x)^2p(y)^3\over{\phi(x)|x-y|}} < 0.
\end{multline}

We now analyze the minimum $a$ and maximum $b$ of the function $H$
(see \eqref{H} and \eqref{Hab}). That $b=1$ is easily seen from the
fact, that $f'(t)>0$ for all $t>0$ which implies that
$F(t):=\int_0^t\rd s\; f(s)\leq f(t)t$. That $a>0$ follows from the
bounds given in Lemma \ref{la1} of Appendix \ref{Anhang1}.  Then (\ref{eq:ungleichung}) is
equivalent to
\begin{equation}
	\label{eq:ungleichung1}
	-3\pi^2 \sum\limits_{k=1}^K Z_k\int_{\rz^3}\rd x\; \frac{H(x)p(x)^3}{\phi(x)|x-R_k|}	+\int_{\rz^3}\rd x\; \int_{\rz^3}\rd y\; {H(x)p(x)^3p(y)^3\over{\phi(x)|x-y|}} < 0.
\end{equation}

This and symmetrizing the second integrand in $x$ allows us to turn
\eqref{eq:ungleichung} into two new inequalities
\begin{equation}
  \label{eq:un1}
  \begin{split}
  &\frac a2 \int_{\rz^3}\rd x\; \int_{\rz^3}\rd y\;
  {\phi(x)^{-1} + \phi(y)^{-1} \over|x-y|}p(x)^3p(y)^3\\
  < &3\pi^2 \sum\limits_{k=1}^K Z_k\int_{\rz^3}\rd x\; g_k(x)H(x)p(x)^3
  \end{split}
\end{equation}
and
\begin{equation}
	\begin{split}
  \label{eq:un2}
  &\frac1{2} \int_{\rz^3}\rd x\; \int_{\rz^3}\rd y\;
  {\phi(x)^{-1} + \phi(y)^{-1} \over|x-y|}H(x) p(x)^3H(y)p(y)^3\\
  <& 3\pi^2 \sum\limits_{k=1}^K Z_k\int_{\rz^3}\rd x\; g_k(x)H(x)p(x)^3,
\end{split}
\end{equation}
where $g_k(x):=(\phi(x)|x-R_k|)^{-1}$. Following an idea of
Baumgartner \cite{Baumgartner1983} and using the triangular
inequality, we have
\begin{equation}
	\label{phi1}
{\phi(x)^{-1} + \phi(y)^{-1} \over|x-y|}=\sum\limits_{k=1}^K \kappa_k\frac{|x-R_k|+|y-R_k|}{|x-y|}g_k(x)g_k(y)\ge\sum\limits_{k=1}^K \kappa_k g_k(x)g_k(y).
\end{equation}

Therefore, applying (\ref{phi1}) on the left side of (\ref{eq:un1})
and (\ref{eq:un2}) yields
\begin{eqnarray}
	\label{schranke1}
	&\frac a2\sum\limits_{k=1}^K \kappa_k \left(\int_{\rz^3}\rd x\;
	g_k(x)p(x)^3\right)^2
	< 3\pi^2 \sum\limits_{k=1}^K Z_k\int_{\rz^3}\rd x\; g_k(x)H(x)p(x)^3,\\
	\label{schranke2}
	&\frac1{2}\sum\limits_{k=1}^K \kappa_k \left(\int_{\rz^3}\rd x\;
	g_k(x)H(x)p(x)^3\right)^2
	< 3\pi^2 \sum\limits_{k=1}^K Z_k\int_{\rz^3}\rd x\; g_k(x)H(x)p(x)^3.
\end{eqnarray}
Adding the two inequalities, we have
\begin{equation}
	\begin{aligned}
	\label{schranke3}
	&6\pi^2 \sum\limits_{k=1}^K Z_k\int_{\rz^3}\rd x\; g_k(x)H(x)p(x)^3\\
&>	\sum\limits_{k=1}^K \kappa_k \left(\frac a2\left(\int_{\rz^3}\rd x\;
g_k(x)p(x)^3\right)^2+\frac1{2}\left(\int_{\rz^3}\rd x\;
g_k(x)H(x)p(x)^3\right)^2\right)\\
&\ge	\sum\limits_{k=1}^K \kappa_k \sqrt{a}\int_{\rz^3}\rd x\;
g_k(x)p(x)^3\int_{\rz^3}\rd x\;
g_k(x)H(x)p(x)^3.\\
\end{aligned}
\end{equation}
This is equivalent to
\begin{equation}
		\label{schranke4}
		 \sum\limits_{k=1}^K \int_{\rz^3}\rd x\; g_k(x)H(x)p(x)^3\left(6\pi^2Z_k-\sqrt{a}\kappa_k \int_{\rz^3}\rd x\;
		g_k(x)p(x)^3\right)> 0.\\
\end{equation}
Following Lieb's \cite{Lieb1984} setting, let
\begin{equation}
	\begin{aligned}
		\label{deltanu}
		&\delta_k:=\frac{\kappa_k}{3\pi^2 N} \int_{\rz^3}\rd x\; g_k(x)p(x)^3,\\
		&\nu_k:=Z_k/Z.
	\end{aligned}
\end{equation}
Note that
\begin{equation}
		\label{deltanu1}
		\sum\limits_{k=1}^K\delta_k=\sum\limits_{k=1}^K\nu_k=1.
\end{equation}
It is proved by Lieb \cite[Appendix B]{Lieb1984} that we can choose {$\kappa_k$} such that
\begin{equation}
	\label{deltanu2}
	\delta_k=\nu_k,\ k=1,...,K.
\end{equation}
Then the left hand side of (\ref{schranke4}) becomes
\begin{equation}
	\label{schranke5}
	3\pi^2(2 Z-\sqrt{a}N)\sum\limits_{k=1}^K \int_{\rz^3}\rd x\; g_k(x)H(x)p(x)^3\delta_k> 0.\\
\end{equation}
The sum is positive, so this yields
\begin{equation}
  N< \frac {2}{\sqrt{a}} Z
\end{equation}
or, numerically $ N < 2.557211758 Z$. Recall that by definition
$a:=\inf H(\rz_+)$ (see \eqref{H} and the line below) and thus depends
only on $f$.
\end{proof}

\appendix

\section{On the semiclassical nature of the Thomas-Fermi term\label{Anhang0}}
Note that the Thomas-Fermi $\cTF$ can be rewritten to emphasize its
semiclassical nature
\begin{equation}
  \label{halbklassisch}
  \cTF[p]= 2 \int_{\rz^3}\rd x\; \int_{|\xi|<p(x)}\rd\xi\;{\rd \xi\over (2\pi)^3}
  (\sqrt{|\xi|^2+1}-1).
\end{equation}
Writing it this way is not only mere curiosity but also shows the
convexity of $\cTF$ and the bound
\begin{equation}
  \label{eq:masselos}
  {1\over 4\pi^2}\int_{\rz^3}\rd x\;\rho(x)^{4/3}\geq \cTF(\sqrt[3]{3\pi^2\rho}) \geq {1\over 4\pi^2}\int_{\rz^3}\rd x\;\rho(x)^{4/3}-\int_{\rz^3}\rd x\;\rho(x).
\end{equation}

This also allows to read off its derivative:
\begin{equation}
  \label{ableitung-p}
  (d\cTF)(p)(\eta) = \frac1{\pi^2}\int_{\rz^3}\rd x\; p(x)^2(\sqrt{p(x)^2+1}-1)\eta(x)
\end{equation}
with $\eta\in L^4(\rz^3)$. Obviously, $(d\cTF)(p)$ can be identified
with an element in $L^{4/3}(\rz^3)$.

Writing the TF functional in $\rho$ yields
\begin{equation}
  \label{ableitung-rho}
  (d\cTF\circ\sqrt[3]{3\pi^2\cdot})(\rho)(\eta)
  = \int_{\rz^3}\rd x\;(\sqrt{(3\pi^2\rho(x))^{2/3}+1}-1)\eta(x).
\end{equation}
for $\eta\in L^{4/3}(\rz^3)$. Obviously this derivative can be
identified with an element in $L^4(\rz^3)$.

\section{Bound on the function $F$, its derivative, and the
  Thomas-Fermi energy\label{Anhang1}}
\begin{lemma}
  \label{la1}
  For all $s,t>0$ 
  \begin{eqnarray}
    \label{a1}
    F(s)&>& \tilde F(s) := s\sqrt{\arsinh(s)}/2,\\
    \label{a2}
    f(s)&>&\tilde f(s):=\sqrt{\arsinh(s)},\\
    \label{a3}
    \tf(s)&>& \tilde\tf(s):={s^4\over1+1/(\tfrac45s)},\\
\label{G}
    \tilde F(s)&>& G(s):= {s^{3/2}\over2\sqrt{1+s}},\\
\label{Ginverse}
G^{-1}(t)&<& 2t^{2/3}\sqrt[3]{t + 1}. 
  \end{eqnarray}
Moreover, $\tilde f$ is monotone increasing and concave.
\end{lemma}
\begin{proof}
  The monotony and concavity of $\tilde f(s)$ follows immediately by
  taking the derivatives
  $$\frac{\rd}{\rd s}\tilde f(s)=\frac{1}{2\sqrt{(s^2+1)\arsinh(s)}}>0$$
  and
  $$\frac{\rd^2}{\rd s^2}\tilde f(s)
  =-\frac{s}{2\sqrt{(s^2+1)^3\arsinh(s)}}-\frac{1}{4(s^2+1)\arsinh(s)^{3/2}}<0.$$
 
 The inequality \eqref{a2} is equivalent to
  $${s\over\sqrt{1+s^2}} +{2s^2\over1+s^2}\arsinh(s)> \arsinh(s)$$
  or
  $$s\sqrt{1+s^2} > (1-s^2)\arsinh(s).$$
  This, however, is clear, since $s>\arsinh(s)$.

  The claim \eqref{a1} follows from \eqref{a2} and the concavity of
  $\tilde f(s)$. We have
 \begin{multline}\notag
   F(t)=\int_0^t\rd s\; f(s)>\int_0^t\rd s\;\tilde f(s)=\int_0^{t/2}\rd s\left(\tilde f(s)+\tilde f(t-s)\right)\\
   >\int_0^{t/2}\rd s\left(\tilde f(0)+\tilde
     f(t)\right)=t\sqrt{\arsinh(t)}/2.
 \end{multline}

 Next we treat \eqref{a3}. To this end we remark that
 $\tf(0)=\tilde\tf(0)=0$. Moreover, the derivative of the difference
 is
  \begin{multline}\notag
    \frac{\rd}{\rd s}\left(\tf(s)-\tilde\tf(s)\right)\\
    =\frac{128\left(\left(p+\frac{5}{4}\right)^2(p^2+1)-\left(\frac{1}{2}p^3+\frac{57}{32}p^2+\frac{5}{2}p+\frac{25}{16}\right)\sqrt{p^2+1}\right)p^2}{(4p+5)^2\sqrt{p^2+1}}>0,
  \end{multline}
  which proves the result.
  
  The inequality \eqref{G} is equivalent to the fact
  $\arsinh(s)>\frac{s}{1+s}$. Using the monotony of $G(s)$ and
  $G(2t^{2/3}\sqrt[3]{t +
    1})=\frac{2^{3/2}t(t+1)^{1/2}}{2\sqrt{1+2t^{2/3}\sqrt[3]{t +
        1}}}>t=G(G^{-1}(t))$, the inequality \eqref{Ginverse} follows.
\end{proof}

\section*{Acknowledgments}
H.C. acknowledges support by the China Scholarship Council through the
LMU-CSC Scholarship Program. H.S. acknowledges partial support by the
Deutsche Forschungsgemeinschaft, grant SI 348/15-1 and
EXC-2111-390814868. Both authors acknowledge support of the Institute
for Mathematical Sciences of Singapore University through their
program ``Density Functionals for Many-Particle Systems: Mathematical
Theory and Physical Applications of Effective Equations''.

%\bibliographystyle{plain}
%\bibliography{coulomb}
\def\cprime{$'$}

\end{document}